\documentclass[10pt,conference]{IEEEtran}
\usepackage{cite}
\usepackage{amsmath,amssymb,amsfonts}
\usepackage{graphicx}
\usepackage{textcomp}
\usepackage{xcolor}
\usepackage{multirow}
\usepackage{makecell}
\usepackage{algorithm,algcompatible}
\usepackage{subfigure}
\usepackage{hhline}
\usepackage{graphicx}
\usepackage{multicol}
\usepackage{graphicx}
\usepackage{amssymb}
\usepackage{amsmath}
\usepackage{amsthm}
\usepackage{siunitx}
\DeclareMathOperator*{\argmax}{arg\,max}

\usepackage{calc}
\usepackage{array}
\usepackage{optidef}
\usepackage{epstopdf}
\usepackage{mathtools}
\usepackage{mathrsfs}
\usepackage{bm}
\usepackage{cite}
\usepackage{boldline}
\usepackage{algorithm}
\usepackage{algpseudocode}

\def\BibTeX{{\rm B\kern-.05em{\sc i\kern-.025em b}\kern-.08em
    T\kern-.1667em\lower.7ex\hbox{E}\kern-.125emX}}

\long\def\comment#1{}

\newfont{\bbb}{msbm10 scaled 700}

\newfont{\bb}{msbm10 scaled 1100}
\newcommand{\CC}{\mbox{\bb C}}

\newcommand{\EE}{\mbox{\bb E}}

\newcommand{\av}{{\bf a}}

\newcommand{\dv}{{\bf d}}

\newcommand{\hv}{{\bf h}}

\newcommand{\pv}{{\bf p}}

\newcommand{\sv}{{\bf s}}

\newcommand{\uv}{{\bf u}}

\newcommand{\xv}{{\bf x}}
\newcommand{\yv}{{\bf y}}
\newcommand{\zv}{{\bf z}}

\newcommand{\Cm}{{\bf C}}

\newcommand{\Fm}{{\bf F}}

\newcommand{\Id}{{\bf I}}

\newcommand{\Um}{{\bf U}}

\newcommand{\Xm}{{\bf X}}

\newcommand{\Ac}{{\cal A}}

\newcommand{\Cc}{{\cal C}}
\newcommand{\Dc}{{\cal D}}

\newcommand{\Nc}{{\cal N}}

\newcommand{\Rc}{{\cal R}}

\newcommand{\eqdef}{\stackrel{\Delta}{=}}

\newcommand{\herm}{{\sf H}}

\newcommand{\transp}{{\sf T}}

\newtheorem{lemma}{Lemma}

\begin{document}
\title{CSIT-Free Downlink Transmission for mmWave MU-MISO Systems in High-Mobility Scenario}

\author{\IEEEauthorblockN{${\textrm{Jeongjae Lee}}$, ${\textrm{Wonseok Choi}}$, and ${\textrm{Songnam Hong}}$
}
\IEEEauthorblockA{
${\textrm{Department of Electronic Engineering, Hanyang University, Seoul, South Korea}}$
\\
Email: 
${\textrm{lyjcje7466@hanyang.ac.kr}}$,
${\textrm{ryan4975@hanyang.ac.kr}}$,
${\textrm{snhong@hanyang.ac.kr}}$
}
}



\maketitle

\begin{abstract}
This paper investigates the downlink (DL) transmission in millimeter-wave (mmWave) multi-user multiple-input single-output (MU-MISO) systems especially focusing on a high speed mobile scenario. To complete the DL transmission within an extremely short channel coherence time, we propose a novel DL transmission framework that eliminates the need for channel state information at the transmitter (CSIT), of which acquisition process requires a substantial overhead, instead fully exploiting the given channel coherence time. Harnessing the characteristic of mmWave channel and uniquely designed CSIT-free unitary precoding, we propose a symbol detection method along with the simultaneous CSI at the receiver (CSIR) and Doppler shift estimation method to completely cancel the interferences while achieving a full combining gain. Via simulations, we demonstrate the effectiveness of the proposed method comparing with the existing baselines.
\end{abstract}


\section{Introduction}
\label{sec:intro}

Millimeter-wave (mmWave) multi-user multiple-input single-output (MU-MISO) systems are a key enabler for achieving  high spectral efficiency. However, in high-mobility scenarios, severe Doppler effects reduce the channel coherence time, making conventional downlink (DL) frameworks that relay on channel estimate information at the transmitter (CSIT) impractical. Acquiring CSIT via uplink pilots or feedback introduces substantial training overhead and latency, fundamentally limiting the applicability of CSIT-based precoding for ultra-reliable low-latency communication (URLLC).


A promising alternative is CSIT-free transmission, which shifts the requirement for channel knowledge to the receiver. Existing approaches, such as space–time coding (STC) \cite{Alamouti1998,Tarokh1999}, offer robustness without CSIT but are largely limited to single-user scenarios and cannot guarantee interference-free multi-user transmission. Moreover, conventional interference-mitigation strategies either rely on downlink channel reconstruction \cite{Kim2024} or assume multi-antenna receivers \cite{Arti2020}, rendering them unsuitable for lightweight devices (e.g., Internet of Things (IoT) devices) in high-mobility networks.

In this paper, we propose a novel CSIT-free DL transmission framework for mmWave MU-MISO systems, particularly focusing on high-mobility scenarios. The key idea is to exploit the orthogonality of circulant permutation discrete Fourier transform (CP-DFT) matrices to design unitary precoding across sequential OFDM data blocks. This approach enables (i) complete elimination of inter-user interference and (ii) full combining gain at each receiver, without requiring CSIT. Furthermore, we develop a user-side symbol detection, and joint channel and Doppler estimation method that leverage only two pilot symbols, ensuring robustness under rapid channel variations. Simulation results under realistic mmWave settings demonstrate that the proposed scheme significantly outperforms both conventional CSIT-based and CSIT-free baselines, especially when the channel coherence time is short. To the best of our knowledge, this is the first framework to achieve interference-free, Doppler-resilient, and CSIT-free multi-user downlink transmission in mmWave systems.

{\em Notations.} Let $[N]\eqdef \{1,2,...,N\}$ for any positive integer $N$. For an $M\times N$ matrix $\Xm$, let $\Xm(i,:)$ and $\Xm(:,j)$ denote its $i$-th row and $j$-th column, respectively. Also, $\Xm^{*}$ and $\Xm^{\herm}$ represent the complex conjugate and Hermitian transpose of $\Xm$. For an $M$-dimensional column vector $\xv$, $\xv([m])$ denotes the sub-vector containing its first $m\in[M]$ components, and $\mbox{diag}(\xv)$ is an $M\times M$ diagonal matrix with the elements of $\xv$ on the diagonal. $\Id_M$ and ${\bf 1}_M$ denote the $M\times M$ identity matrix and the all-one $M$-dimensional column vector.

\section{System Model}
\label{sec:systemmodel}
We consider a DL transmission for mmWave MU-MISO system, where a base station (BS) equipped with $N$ antennas serves $K$ single-antenna users. Each user moves at a high constant velocity $v_k$ with a heading angle $\phi_k\in[0,2\pi)$ for $k\in[K]$ in a single-cell, high-mobility scenario. The BS is located at the center of the cell and employs orthogonal frequency division multiplexing (OFDM). Let $L$ and $L_{\rm CP}$ denote the number of sub-carriers and the length of the cyclic prefix, respectively, and let $B$ represent the system bandwidth. In the baseband (i.e., discrete frequency domain), the channel response from the BS to user $k$ on a sub-carrier $\ell$ follows a Rician fading channel model: 
\begin{equation}
    \sqrt{\frac{\kappa}{\kappa+1}}\hv_{k,\ell}^{\rm LoS} + \sqrt{\frac{1}{\kappa+1}}\hv_{k,\ell}^{\rm NLoS}\stackrel{(a)}{\approx}\hv_{k,\ell}^{\rm LoS}\in\CC^{N},\label{eq:channelmodel}
\end{equation} where $\kappa>0$ is the Rician $\kappa$-factor, and the line-of-sight (LoS) channel response is given by $\hv_{k,\ell}^{\rm LoS} = g_{k,\ell}\av_\ell(\theta_k)\in\CC^{N}$. Herein, the complex channel gain is normalized as $g_{k,\ell} = e^{j\vartheta_{k,\ell}}/r_k\in\CC$, where $\vartheta_{k,\ell}\in[0,2\pi)$ represents a small-scale fading and $r_k$ is the physical distance between the BS and user $k$. Each component of Non-LoS (NLoS) channel response vector $\hv_{k,\ell}^{\rm NLoS}\in\CC^{N}$ follows the distribution $\Cc\Nc(0,1/r_k)$ and the angle of departure (AoD) from the BS to user $k$ is denoted by $\theta_k\in[0,2\pi)$. The normalized array response vector for AoD $\theta$ is defined as
\begin{equation*}
    \av_{\ell}(\theta) = \left[1,e^{-j\nu_{\ell}d_{\rm c}\sin{\theta}},\dots,e^{-j\nu_{\ell}d_{\rm c}(N-1)\sin{\theta}}\right]^{\transp}\in\CC^{N},
\end{equation*}
where $\nu_\ell$ is the wavenumber of sub-carrier $\ell$. Remarkably, \cite{Priebe2013,Wang2025} showed that in mmWave systems, the approximation in (a) in \eqref{eq:channelmodel} is reasonable, as significant reflection losses cause the NLoS components to be attenuated by more than $10$ dB compared to the LoS component. In this paper, therefore, we assume $\kappa\gg 1$ for the design of the transmission scheme, while the effect of NLoS paths (i.e., Rician $\kappa$-factor) is further evaluated through simulations.

\begin{figure}[t]
\centering
\includegraphics[width=0.95\linewidth]{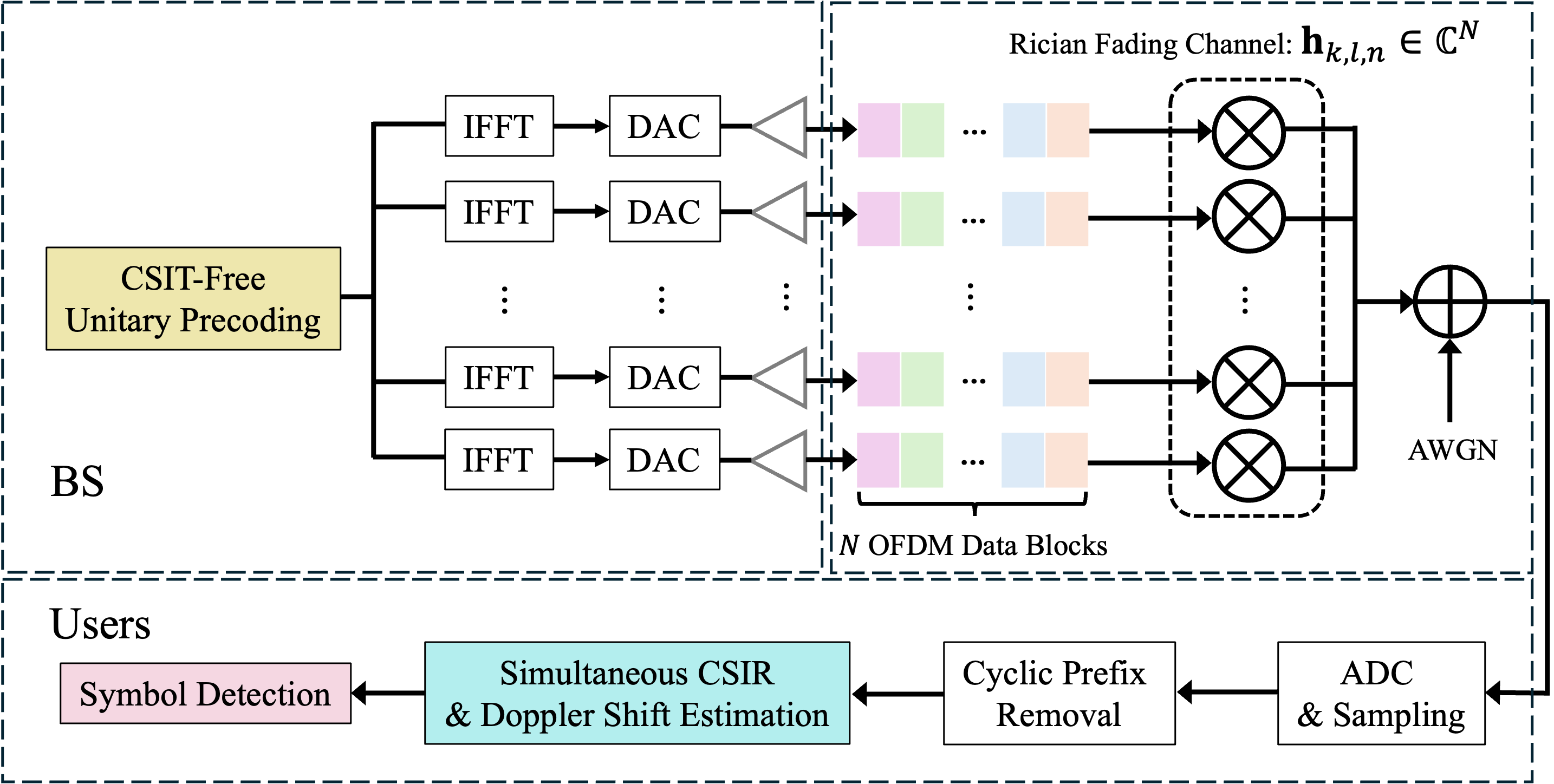}
\caption{Illustration of the proposed CSIT-free downlink transmission framework for MU-MISO mmWave systems.}
\end{figure}

\section{CSIT-Free DL Transmission Framework}\label{sec:framework}

In the high-mobility system model, the channel coherence time, denoted by $T_{\rm c}$, is extremely short due to large Doppler shifts \cite{Liu2021}. Under such conditions, the CSIT-based DL transmission becomes impractical, as the acquiring CSIT—via uplink pilots or feedback—incurs a substantial training overhead. To address this, we propose a DL transmission framework that eliminates the need for CSIT acquisition. Fig. 2 shows the proposed DL transmission framework, where the BS transmits $N$ sequential OFDM data blocks. Each block has a symbol duration of $T_{\rm b} = (L+L_{\rm CP})/B$ and the transmission satisfies the latency constraint $NT_{\rm b}\leq T_{\rm c}$.
The detailed operation of the proposed framework is presented in the subsequent subsections.

\subsection{CSIT-Free Unitary Precoding}
\label{sec:precoding}

We propose a deterministic (CSIT-free) unitary precoding scheme to ensure robustness against noise without increasing the transmit power, while simultaneously guaranteeing interference management so that each user can correctly detect the desired symbol. To achieve both objectives in a CSIT-free system, we exploit a key property of the CP-DFT matrices, as stated in the following lemma:

\begin{lemma}
    Let $\Um$ be an $N\times N$ normalized DFT matrix. The $n$-th CP-DFT matrix is defined as: 
    \begin{equation*}
         \Um_{n} \eqdef \left[
            \Um(:,\Cm(1,n)) \;\cdots\; \Um(:,\Cm(N,n))
        \right]\in\CC^{N\times N},
    \end{equation*} where $\Cm$ is the $N\times N$ circulant matrix with the first column $\begin{bmatrix}
            1 & 2 &\cdots & N 
       \end{bmatrix}^{\transp}$.
    Then, it holds that $\Um_n\Um_{n'}^{\herm}=\Id_N$ for $n=n'$, and $\Um_n\Um_{n'}^{\herm}=\mbox{\em diag}(\uv_{[n,n']})$ for $n\neq n'$, where $\uv_{n,n'}$ is a vector satisfying ${\bf 1}_N^{\herm}\uv_{n,n'} = 0$.
\end{lemma}
\begin{proof}
    The proof is given in {\cite[Lemma 1]{Lee2025Blind}}.
\end{proof} 

Lemma 1 establishes the orthogonality between different CP-DFT matrices. By harnessing this property, the BS precodes the symbols of a sub-carrier $\ell$ in data block $n\in[N]$ as follows:
\begin{equation}
    \xv_{\ell,n} = \frac{1}{\sqrt{N}}\Fm_n\mbox{diag}(\pv_\ell)\sv_{\ell}\in\CC^{N},
\end{equation} where the precoding matrix is defined as: 
\begin{equation}
    \Fm_{n} = \left[
       \Um_1(:,n) \;   \Um_2(:,n) \; \cdots \; \Um_N(:,n)\right]\in\CC^{N\times N},
\end{equation} and
$\pv_{\ell}=[\sqrt{p_1},\sqrt{p_2},\dots,\sqrt{p_N}]\in\CC^{N}$ denotes the power allocation vector. The vector $\sv_{\ell}\in\CC^{N}$ contains $K$ information symbols for all users in $\sv_{\ell}([K])$, and two pilot symbols in $\sv_{\ell}(N-1)$ and  $\sv_{\ell}(N)$, such that $\EE[\sv_\ell\sv_\ell^{\herm}]=\Id_N$. The two pilot symbols are required for CSI at the receiver (CSIR) and for Doppler shift estimation at each user. Consequently, the BS can serve $K=N-2$ users under the proposed data transmission framework. During $N$ data blocks, symbols are precoded using different CP-DFT matrices, which guarantee orthogonality among the transmitted symbols, thereby eliminating inter-user interference.

\subsection{Symbol Detection}
After cyclic prefix removal is performed using the method in \cite{Wang2018} for the quantized and sampled signals in data block $n$, the $n$-th received signal at user $k$ on sub-carrier $\ell$ can be expressed as $y_{k,\ell,n} = \hv_{k,\ell,n}^{\transp}\xv_{\ell,n} + z_{k,\ell,n}\in\CC$, where the DL channel response at data block $n$, accounting for the Doppler shift across blocks, is given by:
\begin{equation}
    \hv_{k,\ell,n} = e^{j\nu_{\ell}(n-1)T_{\rm b}v_k\cos(\phi_k-\theta_k)}\hv_{k,\ell}^{\rm LoS},
\end{equation} and $z_{k,\ell,n}\sim \Cc\Nc(0,\sigma^2)$ denotes the additive white Gaussian noise (AWGN) with power $\sigma^2$. Concatenating the received signals over $N$ data blocks, user $k$ forms the observation vector $ \yv_{k,\ell} = [y_{k,\ell,1},y_{k,\ell,2},\dots,y_{k,\ell,N}]^{\transp}\in\CC^{N}$.

Let $\hat{g}_{k,\ell}\in\CC$ and $\hat{\theta}_k\in[0,2\pi)$ denote the estimated complex channel gain and AoD, respectively. User $k$ then detects the symbol $s_{\ell}(n)$ by linearly combining the observation vector with the corresponding CP-DFT matrix $\Um_n$ as follows:
\begin{align}
    c_{k,\ell,n}(\hat{g}_{k,\ell},\hat{\theta}_k) &= \frac{\av_{\ell}^{\herm}(\hat{\theta}_k)\Um_n^{*}\mbox{diag}(\dv_{\ell}^{*}(\hat{\theta}_k))\yv_{k,\ell}}{\hat{g}_{k,\ell}}\in\CC,\label{eq:classification}
\end{align} where $\dv_\ell(\theta) = \begin{bmatrix}
        1,\dots,e^{j\nu_{\ell}(N-1)T_{\rm b}v_k\cos(\phi_k-\theta)}
    \end{bmatrix}^{\transp}\in\CC^{N}$ denotes the Doppler shift vector. The expression in \eqref{eq:classification} can be decomposed into three additive terms:
\begin{align}
    &c_{k,\ell,n}(\hat{g}_{k,\ell},\hat{\theta}_k) = \overbrace{\frac{{g}_{k,\ell}u_{k,\ell,n}(\hat{\theta}_k)s_{\ell}(n)\sqrt{{p_n}}}{\hat{g}_{k,\ell}\sqrt{N}}}^{\rm desired\; symbol\; term}\nonumber\\ 
    &+ \underbrace{\sum_{n'\neq n}^{N}\frac{{g}_{k,\ell}{v_{k,\ell,n'}(\hat{\theta}_k)s_{\ell}(n')}\sqrt{{p_{n'}}}}{\hat{g}_{k,\ell}\sqrt{N}}}_{\rm interference\;term} + \underbrace{\tilde{z}_{k,\ell,n}(\hat{g}_{k,\ell},\hat{\theta}_k)}_{\rm noise\; term},\label{eq:decoded}
\end{align} where $u_{k,\ell,n}(\hat{\theta}_k)\in\CC$ and $v_{k,\ell,n'}(\hat{\theta}_k)\in\CC$ are defined in \eqref{eq:signal} and \eqref{eq:insignal}, respectively. The noise term is given by:
\begin{equation}
    \tilde{z}_{k,\ell,n}(\hat{g}_{k,\ell},\hat{\theta}_k) =  \frac{\av_{\ell}^{\herm}(\hat{\theta}_k)\Um_n^{*}\mbox{diag}(\dv_{\ell}^{*}(\hat{\theta}_k))\zv_{k,\ell}}{\hat{g}_{k,\ell}}\in\CC,
\end{equation} with $\zv_{k,\ell} = [z_{k,\ell,1},z_{k,\ell,2},\dots,z_{k,\ell,N}]^{\transp}\in\CC^{N}$. Notably, our key result is provided in Lemma~\ref{lem:linear_combiner}:
\begin{lemma}\label{lem:linear_combiner} When $\hat{g}_{k,\ell} = g_{k,\ell}$ and $\hat{\theta}_k = \theta_k$, the combined signal in \eqref{eq:decoded} reduces to 
\begin{equation*}
    c_{k,\ell,n}(\hat{g}_{k,\ell},\hat{\theta}_k) = s_{\ell}(n)\sqrt{Np_{n}} + \tilde{z}_{k,\ell,n}({g}_{k,\ell},{\theta}_k).\label{eq:max}
\end{equation*}
\end{lemma}
\begin{proof} We first show that the combining gain of the desired symbol term equals $N$:
\begin{align}
    u_{k,\ell,n}(\hat{\theta}_k) &= \av_\ell^{\herm}(\hat{\theta}_k)\Um_n^{*}\mbox{diag}(\dv_\ell^{*}(\hat{\theta}_k)\circ\dv_\ell({\theta}_k))\Um_n^{\transp}\av_\ell(\theta_k)\nonumber\\
    &=\av_\ell^{\herm}({\theta}_k)\left(\Um_n\Um_n^{\herm}\right)^{*}\av_\ell(\theta_k)\stackrel{(a)}{=}N,\label{eq:signal}
\end{align} where $\circ$ denotes the Hadamard product, and (a) follows from Lemma 1. 

Next, we show that the interference term is completely canceled:
\begin{align}
    v_{k,\ell,n'}(\hat{\theta}_k) &= \av_\ell^{\herm}(\hat{\theta}_k)\Um_n^{*}\mbox{diag}(\dv_\ell^{*}(\hat{\theta}_k)\circ\dv_\ell({\theta}_k))\Um_{n'}^{\transp}\av_\ell(\theta_k)\nonumber\\
    &=\av_\ell^{\herm}({\theta}_k)\left(\Um_n\Um_{n'}^{\herm}\right)^{*}\av_\ell(\theta_k)\nonumber\\
    &\stackrel{(a)}{=}\av_\ell^{\herm}({\theta}_k)\mbox{diag}(\uv_{n,n'})^{*}\av_\ell(\theta_k)
    \stackrel{(b)}{=}0,\label{eq:insignal}
\end{align} where (a) follows from Lemma 1, and (b) holds because $\av_\ell^{\herm}({\theta}_k)\mbox{diag}(\av_\ell(\theta_k)) = {\bf 1}_N^{\transp}$ and ${\bf 1}_{N}^{\herm}\uv_{n,n'} = ({\bf 1}_{N}^{\transp}\uv_{n,n'}^{*})^{*}=0$ as given by Lemma 1. This completes the proof.
\end{proof} 

From Lemma~2, the symbol can be accurately detected from the combined signal in \eqref{eq:decoded}, provided that the CSIR and Doppler shift are precisely estimated. Therefore, obtaining accurate estimates of these parameters is essential before performing the proposed symbol detection.


\section{CSIR and Doppler Shift Estimation}

Focusing on user $k$, we describe the proposed CSIR and Doppler shift estimation method. The same procedure is subsequently applied to the other users. We first define the signal-to-noise-plus-interference ratio (SINR) from the combined signal in \eqref{eq:decoded} as:
\begin{equation}
     \gamma_{k,\ell,n}(g,\theta) = \frac{|{\rm desired\;symbol\;term}|^2}{|{\rm interference\;term}|^2 + |{\rm noise\; term}|^2},\label{eq:SINR}
\end{equation} where $g\in\CC$ and $\theta\in[0,2\pi)$ denote arbitrary complex channel gain and AoD, respectively. Based on this definition, a key property enabling CSIR and Doppler shift estimation is provided in Lemma~3:
\begin{lemma} 
    For any complex channel gain $g\in\CC$, we let
    \begin{equation*}
        \theta^{\star} = \argmax_{\theta\in[-\frac{\pi}{2},\frac{\pi}{2})}\;\gamma_{k,\ell,n}(g,\theta).
    \end{equation*} Then, it holds that $\av_{\ell}(\theta^{\star}) = \av_{\ell}(\theta_k)$ and $\dv_{\ell}(\theta^{\star}) = \dv_{\ell}(\theta_k)$.
\end{lemma}
\begin{proof}
    The proof is provided in in \cite[Corollary 1]{Lee2025Blind}.
\end{proof} Leveraging Lemma 3, we now describe the estimation of the CSIR and Doppler shift using the two pilot symbols, $\sv_{\ell}({N-1})$ and $\sv_{\ell}(N)$. Through quantization, we respectively construct codebooks for the array response and Doppler shift vectors as $\Ac_{\ell}=\{\av_\ell(\Delta_q):q\in[Q]\}$, and $\Dc_{\ell}=\{\dv_\ell(\Delta_q):q\in[Q]\}$, where $Q$ is the quantization resolution and $\Delta_q = -{\pi}/{2} + {(q-1)\pi}/{Q}$ for $q\in[Q]$. From Lemma 2, and assuming that $\av_\ell(\Delta_q)$ and $\dv_\ell(\Delta_q)$ correspond to the true array response and the Doppler shift vectors, respectively (i.e., $\Delta_q  = \theta_k$), the combined signal in \eqref{eq:decoded} for $n=N-1$ can be expressed as:
\begin{align*}
    c_{k,\ell,N-1}(1,\Delta_q) \stackrel{(a)}{=}g_{k,\ell}\sv_{\ell}(N-1)\sqrt{Np_{n}}+\tilde{z}_{k,\ell,n}(1,\Delta_q).
\end{align*} 
From this, the complex channel gain scaled by $\sqrt{p_{n}}$ for the $q$-th AoD is estimated as:
\begin{equation}
    \hat{g}_{k,\ell,q}= \frac{c_{k,\ell,N-1}(1,\Delta_q)}{\sqrt{N}\sv_{\ell}(N-1)}.
\end{equation} It is important to note that this estimator performs well only when the assumption $\Delta_q  = \theta_k$ holds. For AoDs deviating from the true direction, the resulting SINR significantly degrades, which can be exploited to identify the desired quantized AoD, as discussed follows. 

Using the estimated complex channel gain for the $q$-th AoD, i.e., $\hat{g}_{k,\ell,q}$, the combined signal in \eqref{eq:decoded} for $n=N$ can be expressed as:
\begin{align}
    c_{k,\ell,N}(\hat{g}_{k,\ell,q},\Delta_q) \stackrel{(a)}{=} \sqrt{N}\sv_{\ell}(N)+\tilde{z}_{k,\ell,n}(\hat{g}_{k,\ell,q},\Delta_q).\label{eq:ICLASS}
\end{align} where (a) follows from Lemma 2 under the assumptions that $\Delta_q = \theta_k$ and the complex channel gain is perfectly estimated. From \eqref{eq:SINR} and \eqref{eq:ICLASS}, the user computes the SINR as:
\begin{equation*}
    \hat{\gamma}_{k,\ell,N}(\hat{g}_{k,\ell,q},\Delta_q) = \frac{|\sqrt{N}\sv_{\ell}(N)|^2}{|c_{k,\ell,N}(\hat{g}_{k,\ell,q},\Delta_q) - \sqrt{N}\sv_{\ell}(N)|^2},\label{eq:estSINR}
\end{equation*} and then computes the estimated mean spectral efficiency as $\hat{\Rc}_{k,q} = (1/L)\sum_{\ell=1}^{L}\log_2\left(1+\hat{\gamma}_{k,\ell,N}(\hat{g}_{k,\ell,q},\Delta_q)\right)$. Finally, the AoD index that maximizes the mean spectral efficiency is selected as:
\begin{equation}
    q_k^{\star} = \argmax_{q\in[Q]}\;\hat{\Rc}_{k,q},
\end{equation} and the corresponding estimated complex channel gain and AoD are denoted by $\hat{g}_{k,\ell,q_k^{\star}}$ and $\Delta_{q_k^{\star}}$, respectively. 

To analyze the computational complexity of the proposed method, we evaluate the number of complex multiplications required by the proposed linear combiner for CSIR and Doppler shift estimation, where the total number of complex multiplications required to operate the proposed method is computed as $(2L+LQ)N^2 + LQN$.

\section{Simulation Results}

We consider a mmWave MU-MISO system in a high-mobility scenario with parameters: $f_{\rm c} = 28$ GHz, $B = 6.25$ MHz, $L =\num{64}$, and $L_{\rm CP} = 16$. To satisfy the latency constraint, i.e., $NT_{\rm b}\leq T_{\rm c}$, we set $N = 10$, $K = 8$, and $v_k = 39$ m/s, $\forall k\in[K]$. To evaluate the effect of the Doppler shift estimation, we simulate the worst-case scenario in which the Doppler shift is maximized, i.e., $\phi_k = \theta_k$, $\forall k\in[K]$. The cell coverage is $200$ m, the noise power is $\sigma^2 = -30$ dBm, and the quantization resolution is $Q = 256$. Monte Carlo simulations with $10^3$ trials are performed, where in each trial the distances between the BS and users, $\{r_k:k\in[K]\}$, are uniformly and randomly chosen from $[100, 200]$ m. An equal power allocation is assumed, i.e., $p_1= p_2=\dots=p_N$, to evaluate the average performance across users at different distances. We compare the proposed method with the following schemes: i) the performance limit of the proposed method under perfect CSIR and Doppler shift, ii) R-CIRCLE \cite{Lee2025Blind}, iii) zero-forcing (ZF), and iv) maximum ratio transmission (MRT). For the CSIT-based methods (ZF and MRT), perfect CSIT at the BS is assumed, with power allocation identical to that of an OFDM data block in the proposed CSIT-free DL transmission.

\begin{figure}[t]
\centering
\includegraphics[width=0.95\linewidth]{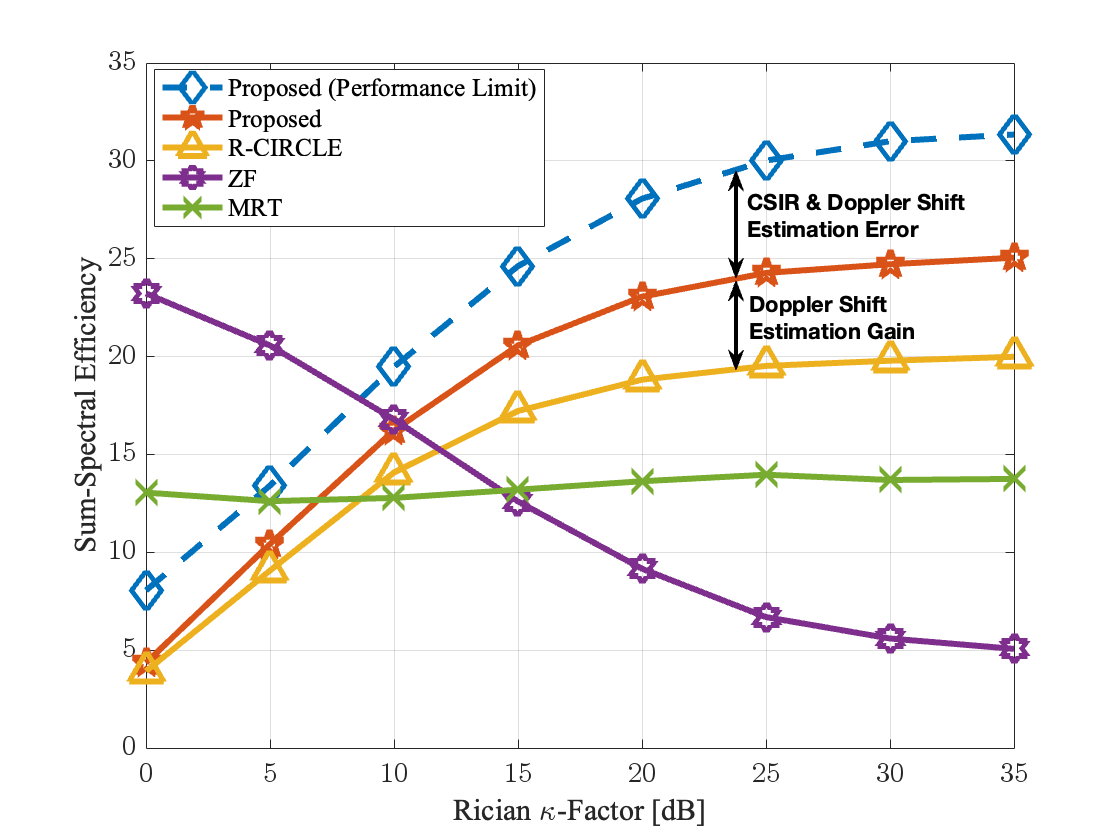}
\caption{Sum-spectral efficiency on the Rician $\kappa$-factor. $p_n=25$ dBm, $\forall n\in[N]$.}
\end{figure}


Fig. 2 illustrates the impact of the Rician $\kappa$-factor on the sum-spectral efficiency. The proposed scheme improves consistently as $\kappa$ increases and the LoS component dominates. While performance degrades slightly at low $\kappa$ due to stronger NLoS contributions, in the practical mmWave regime (i.e., $\kappa \gg 1$) the proposed framework achieves near-optimal performance by fully exploiting the orthogonality of CP-DFT precoding to guarantee interference cancellation and full combining gain. Compared with R-CIRCLE \cite{Lee2025Blind}, the proposed method yields a clear performance advantage as it encompasses Doppler shift estimation, thereby guaranteeing robust symbol detection even under high mobility. Furthermore, the proposed scheme avoids the need for CSIT acquisition. Unlike CSIT-based methods, which require frequent feedback and channel updates within the short coherence time, the proposed approach performs joint channel and Doppler estimation only after receiving all signals within the coherence time. This design makes it particularly effective in ultra-fast mobility scenarios, offering superior reliability and efficiency without incurring excessive overhead.





\section{Conclusion}
In this paper, we proposed a novel CSIT-free DL transmission framework for mmWave MU-MISO systems in high-mobility scenarios. By exploiting the orthogonality of CP-DFT precoding, the proposed method ensures complete interference cancellation and full combining gain, while requiring only two pilot symbols for joint CSIR and Doppler estimation. Simulation results demonstrated that the proposed scheme achieves near-optimal performance in practical mmWave channels, outperforming both CSIT-based and CSIT-free baselines under short coherence times. These results confirm the effectiveness and practicality of the proposed approach for high-mobility mmWave networks.



\end{document}